\documentclass[runningheads,a4paper]{llncs}

\usepackage{amssymb}
\usepackage{graphicx}

\usepackage{url}
\urldef{\mailsa}\path| mahyuddin@usu.ac.id, nasutionmahyu2012@gmail.com  |
\newcommand{\keywords}[1]{\par\addvspace\baselineskip
\noindent\keywordname\enspace\ignorespaces#1}

\begin{document}

\mainmatter  
\title{Extracted Social Network Mining}

\titlerunning{Extracted Social Network Mining}
\author{Mahyuddin K. M. Nasution
\thanks{Proceeding of International Conference on Information Technology and Engineering Application (5-th ICIBA), 86-91, February 19-20, 2016.}
\authorrunning{Mahyuddin K. M. Nasution}
\institute{Fakultas Ilmu Komputer dan Teknologi Informasi (Fasilkom-TI)\\
Universitas Sumatera Utara, Padang Bulan, Medan 20155, Sumatera Utara, Indonesia\\
\mailsa\\}}
\toctitle{}
\tocauthor{}
\maketitle

\begin{abstract}
In this paper we study the relationship between the resources of social networks
by exploring the Web as big data based on a simple search engine. We have used set theory by utilizing the occurrence and co-occurrence for defining the singleton or doubleton spaces of event in a search engine model, and then provided them as representation of social actors and their relationship in clusters. Thus, there are behaviors of social actors and their relation based on Web.

\end{abstract}

\section{Introduction}
An extracted social network is a resultant from the methods of extracting social network from information sources (web pages, documents, or corpus) \cite{nasution2010} or the transformation of the raw data into a social network (pre-processing) \cite{arasu2003}. However, Web not only dealing with everything changed dynamically \cite{nasution2012a}, but Web as social media represent all members of social (population) \cite{bent2006,abdillah2014}, or containing big data as big picture of world. Thus, extraction of social networks always based on parts of social (communities) \cite{nasution2011}, and then to analyze it so that enable to generate useful information, for example, in the decision making \cite{nizamb2015}. This needs the sample that can represent population. Therefore, for getting significance of information source and trust, in the extracted social networks need the suitable approaches \cite{yao2012}.

In other side, the resources of social network such as vertices/actors, edges/re\-la\-tions, and Web/documents have the relations between one to another \cite{nasution2015}. A lot of relations between vertices and edges for expressing some of social structures in Social Network Analysis (SNA) \cite{scott2000}, but still a little formula to get information about relations among first two resources and Web \cite{memon2010}. Therefore, this needs a formalism study about resources.  In all sides of social network, this paper aimed to provide a basic means of discovering knowledge formally about the extracted social network, we call it social network mining. 

\section{Related Work and Motivation}
The social networks can be modeled naturally by the graph $G\langle V,E\rangle$ where $V = \{v_i|i=1,\dots,n\}$ is a set of vertices, and $\{e_j|j=1,\dots,m\}$ is a set of edges and $e_j \in E$ if two vertices $v_k\in V$ and $v_l\in V$ are adjacent, or $e_j = v_kv_l = v_lv_k$ \cite{nasution2012b}. In pre-processing of social network mining, extracting the social network from the information sources is the relatively approaches which is formed through modal relations \cite{nasution2010}. One of extraction methods is the the superficial method that depends heavily on the occurrence and the co-occurrence \cite{nasution2012a}. 

Let a word "Web" or a phrase "World Wide Web" is representation an object according to what we think \cite{nasution2011a}: the computer network is a social network \cite{wellman1996}. In expressing the behavior of social. Natural language processing (NLP) as basic layer of social network mining, and we define the term related it as follows.

\begin{definition}
A term $t_x$ consists of at least one or a set of words in a pattern, or $t_x = (w_1,\dots,w_l)$, $l\le k$, $k$ is a number of words $w$(s), $l$ is number of vocabularies (tokens) in $t_x$, $|t_x| = k$ is size of $t_x$.
\end{definition}

In NLP, 'Shahrul Azman Noah' and 'Opim Salim Sitompul' as terms, for example, are well-defined names of social actor. We have defined a dynamic space based on concept of NLP application as follows \cite{nasution2012c}.

\begin{definition} 
\label{def:se}
Let a set of web pages indexed by search engine by $\Omega$. For each search term $t_x$, where $t_x\in \Sigma$, i.e. a set of singleton search term of search engine. There are a dynamic space $\Omega$ containing the ordered pair of the term $t_{x_i}$ $i=1,\dots,I$ and web pages $\omega_{x_j}$ $j = 1,\dots,J$: $(t_{x_i},\omega_{x_j}) = (t_x,\omega_x)_{ij}$, or a vector space $\Omega_x = (t_x,\omega_x)_{ij}\subseteq\Omega$ is a singleton search engine event of web pages (\emph{singleton event}) that contain an occurrence (\emph{event}) of $t_x\in\omega$. 
\end{definition}

\begin{definition}
\label{def:impli}
Suppose $t_x\in q$ and $q$ is a query. Clustering web pages based on query is an implication, i.e if $\omega\Rightarrow t_x$ is \textsc{True} then a web page $\omega\in\Omega$ is relevant to $q$ or 
\[
\Omega_x =\cases{1 & if $t_x$ is true at all $\omega\in\Omega$,\cr
0 & otherwise}
\]
and $\Omega_x$ as the cluster of $t_x$. 
\end{definition}

Classically, a logical implication associated with inference \cite{nasution2012d}. 

\begin{lemma}
\label{lem:prob}
If $|\Omega|$ is the cardinality of $\Omega$ and $|\Omega_x|\le|\Omega|$ then probability of a singleton event $\Omega_x$ is
\[
P(t_x) = |\Omega_x|/|\Omega|\in[0,1].
\]
\end{lemma}
\begin{proof}
For any term $t_x \in q$, each web page $\omega \in \Omega$ is relevant to a query $q$ has a probability to other web pages in $\Omega$, $0\le p(\omega)=1/|\Omega|\le 1$. Probability of all web pages that relevant to a query in $\Omega$ is $0\le p(\Omega_x) = \sum p(\omega) = |\Omega_x|/|\Omega|\le 1$, or $P(t_x) = p(\Omega_x)$. 
\end{proof}

In the same concept we have defined also the co-occurrence based on NLP \cite{nasution2012e}. 

\begin{definition}
\label{def:doubleton}
Let $t_x$ and $t_y$ are two different search terms, $t_x\ne t_y$, $t_x,t_y\in\Sigma$, where $\Sigma$ is a set of singleton term of search engine. There are a dynamic space $\Omega$ containing the ordered pair of two terms $\{t_{x_i},t_{y_i}\}$ $i=1,\dots,I$ and web pages $\omega_{x_j}$ $j = 1,\dots,J$: $(\{t_{x_i},t_{y_i}\},\omega_{xy_j}) = (\{t_x,t_y\},\omega_{xy})_{ij}$, or a vector space $\Omega_x\cap\Omega_y = (\{t_x,t_y\},\omega_{xy})_{ij}\subseteq\Omega$ is a doubleton search engine event of web pages (\emph{doubleton event}) that contain a co-occurrence (\emph{event}) of $t_x,t_y\in\omega$. 
\end{definition}

\begin{lemma}
If $|\Omega|$ is the cardinality of $\Omega$ and $|\Omega_x\cap\Omega_y|\le|\Omega|$ then probability of a doubleton event $\Omega_x\cap\Omega_y$ is 
$
P(\{t_x,t_y\}) = |\Omega_x\cap\Omega_y|/|\Omega|\in[0,1].
$
\end{lemma}
\begin{proof} 
As direct consequence of: Definition \ref{def:doubleton} and Lemma \ref{lem:prob}.
\end{proof}

At the time conducting the extraction for getting occurrences and co-occur\-ren\-ce, we submitted the queries containing the name to Google search engine, we have the hit count = 20,000 for 'Shahrul Azman Noah' (as occurrence) and = 3,000 for 'Opim Salim Sitompul' (as occurrence), while the hit count for 'Shahrul Azman Noah,Opim Salim Sitompul' (as co-occurrence) is 218. However, if the query contains names that are enclosed in quotation marks, produced the hit count = 2,680 for "Shahrul Azman Noah" (as occurrence) and the hit count = 5,650 for "Opim Salim Sitompul" (as occurrence), while the hit count for '"Shahrul Azman Noah","Opim Salim Sitompul"' (as co-occurrence) is 61. Therefore, information about social actors in occurrence and social networks in co-occurrences are different in behavior, and we have an assumption \cite{chen2008,wang2011}: Each probability of forming its own distribution. Different data distribution gives different behavior. In this case, we have the problem.

\begin{theorem}
\label{theo:behavior}
The behavior of a cluster describes the behavior of a social actor, then the behavior of other actors expressed by the relationships between the clusters.
\end{theorem}
 
\section{Model and Approach}
Literally, we can identify social actor based on Named-Entity Recognition (NER) in web pages or any document as follow.

\begin{definition}
Suppose there are the well-defined actors, then there is $A = \{a_i|i=1,\dots,n\}$ as a set of social actors.
\end{definition} 

Each actor literally also has attributes, thus we can define it as follow \cite{nasution2014a}.
\begin{definition}
Suppose there be the well-identified attributes, then there is $B = \{b_j|j=1,\dots,m\}$ as a set of attributes of actors.
\end{definition}

\begin{definition}
\label{def:rela}
For all pairs (\emph{dyads}) of $n$ social actors, a set of relationships $R = \{r_p|p=1,\dots,m\}$ where a relationship between two actors there are a tie connect them by one or more relations, or $r_p(a_k,a_l) = B_{a_k}\cap B_{a_l}$.
\end{definition}

\begin{definition}
\label{def:extract}
An extracted social network, i.e. $SN = \langle V,E,$ $A,R,\gamma_1,\gamma_2\rangle$ satisfies the conditions as follow:
\begin{enumerate}
\item $\gamma_1 : A\stackrel{1:1}{\rightarrow} V$, and
\item $\gamma_2 : R\rightarrow E$. 
\end{enumerate}
\end{definition}

As an approach to formalize the relationship between resources of social networks, and for exploring the behavior, we use the association rule.

\begin{definition}
\label{def:ar}
Let $B = \{b_1,b_2,\dots,b_m\}$ is a set of attributes. Let $M_i$ is a set of transactions are subsets of attributes or $M_i\subseteq B$. The implication $\Omega_{b_k}\Rightarrow\Omega_{b_l}$ with two possible value \textsc{True} or \textsc{False} as an \emph{association rule} if $\Omega_{b_k}, \Omega_{b_l}\subset B$ and $\Omega_{b_k}\cap\Omega_{b_l}=\emptyset$
\end{definition}

\section{Formulation of Behavior}

Each cluster represents an actor based on the extraction of social networks.

\begin{lemma}
\label{lem:sac}
If for a cluster $\Omega_x$ of a search term $t_x$ there exist other cluster $\Omega_y$ of a search term $t_y$ where $t_x\ne t_y$, then $\Omega_x$ is a \emph{stand-alone cluster}.
\end{lemma}
\begin{proof}
Based on Definition \ref{def:se} and Definition \ref{def:ar}, we have $t_x \Rightarrow t_y$ literally or $\Omega_x\Rightarrow\Omega_y$, but $t_x\ne t_y$ such that $\Omega_x\cap\Omega_y = \emptyset$. Therefore, $\Omega_x$ is a stand-alone cluster.
\end{proof}

\begin{proposition}
\label{pro:actor}
If $t_{a_i}\in q$ $i=1,\dots,n$ and $\Omega_{a_i}$ are a stand-alone cluster for each of $\{a_1,a_2,\dots,a_n\} = A$, then $\Omega_{a_i}$ represent the behavior of $a_i\in A$, respectively.
\end{proposition}
\begin{proof}
Based on Definition \ref{def:ar}, we have $\omega\in\Omega \Rightarrow t_a\in q$ and $\omega$ is representation of actor $a\in A$, and because of each $\omega \in \Omega$ has a probability then $\omega\in\Omega$ be the behavior of actor $a\in A$, but based on Definiton \ref{def:se} $\Omega_a = \{(t_a,\omega_a)_{ij}\}$, $\Omega_a$ is representatin of $a\in A$. Let there be $t_{a_k},t_{a_l}\in q$ $t_{a_k}\ne t_{a_l}$, we have $\Omega_{a_k}\Rightarrow\Omega$ and $\Omega_{a_l}\Rightarrow\Omega$: Even though $\Omega_{a_k}\Rightarrow\Omega_{a_l}$ or $\Omega_{a_l}\Rightarrow\Omega_{a_k}$, but $\Omega_{a_k}\cap\Omega_{a_l} = \Omega_{a_l}\cap\Omega_{a_k} = \emptyset$. Each of $\Omega_{a_i}, i=1,\dots,n$ is a stand-alone cluster that represent the behavior of an actor.
\end{proof}

\begin{lemma}
\label{lem:sacp}
Let $t_{a_k}\ne t_{a_l}$ is the different search terms represent two social actors. If $t_{a_k},t_{a_l}\in q$, then $\Omega_{a_{kl}}$ is a stand-alone cluster for a pair of social actors.
\end{lemma}
\begin{proof}
As applicable in Lemma \ref{lem:sac} to Definition \ref{def:impli} and Definition \ref{def:se}, $\omega\in\Omega\Rightarrow\{t_{a_k},t_{a_l}\}\in q$ or $\omega\in\Omega\Rightarrow\{t_{a_k}\wedge t_{a_l}\}\in q$ and
$
(\omega\in\Omega\Rightarrow t_{a_k}\in q)\wedge(\omega\in\Omega\Rightarrow t_{a_l}\in q)
$
and we have $\Omega_{a_k}\Rightarrow\Omega_{a_l}$ and $t_{a_k}\ne t_{a_l}$, but $\Omega_{a_k}\cap\Omega_{a_l}\ne\emptyset$ then $\Omega_{a_l}\Rightarrow\Omega_{a_k}$. However, based on Definition \ref{def:ar} we have $((\Omega_{a_k}\Rightarrow\Omega_{a_l})\Rightarrow\Omega) = ((\Omega_{a_l}\Rightarrow\Omega_{a_k})\Rightarrow\Omega)$. In other word, 
$\Omega_{a_{kl}} = \{(t_{a_{kl}},\omega_{a_{kl}})_{ij}\}
                = \{(t_{a_k}\wedge t_{a_l},\omega_{a_k}\wedge\omega_{a_l})_{ij}\}
                = \{(t_a\wedge t_a,\omega_a\wedge\omega_a)_{ij}\} (i,j=k\wedge l)
                = \{(t_a,\omega_a)_{ij}\} (i,j=k\wedge l)
                = \Omega_a$.
Thus, $\Omega_{a_{kl}}$ is a stand-alone cluster of a pair of social actors.
\end{proof}

\begin{proposition}
\label{pro:relation}
If $\Omega_{a_{kl}}$ is a stand-alone cluster for a pair of $\{a_1,a_2,\dots,a_n\} = A$, then $\Omega_{a_kl}$ represent the behavior of relationship between $a_i\in A$, $i=1,\dots,n$.
\end{proposition}
\begin{proof}
Based on Lemma \ref{lem:sacp} we have $\Omega_{a_{kl}} = \Omega_a$, and 
$
\Omega_a = \{(t_{a},\omega_a)_{ij}\}
         = \{(t_{a}\wedge t_{a},\omega_a\wedge\omega_a)_{ij}\}
         = \{(t_{a_k}\wedge t_{a_l},\omega_{a_k}\wedge\omega_{a_l})\} (i,j=k\wedge l)
         = \{(t_{a_k}\wedge t_{a_l},\omega_{a_k}\wedge\omega_{a_l})\}
         = \{(t_{a_k},\omega_{a_k})\wedge (t_{a_l},\omega_{a_l})\}
         = \{(t_{a_k},\omega_{a_k})\}\cap\{(t_{a_l},\omega_{a_l})\}
         = \Omega_{a_k}\cap\Omega_{a_l}.$
Or because name also can be an attribute of social actor, then Based on Definition \ref{def:rela} we have
$
\Omega_{a_k}\cap\Omega_{a_l} = B_{a_k}\cap B_{a_l}
         = r_p(a_k,a_l).$

\end{proof}

Definition \ref{def:extract} has set the existence of a social actor by means of $\gamma_1$ and behavior of a social actor based on the result clusters (Proposition \ref{pro:actor}), while the behavior of relationship between social actors refers to the cluster based on dyad (Proposition \ref{pro:relation}) and this behavior based on $\gamma_2$ also become behavior of an edge in social network. Specially, in superficial methods $r_p\in R$ means the strength relation between two actors $a_k$ and $a_l$ in $A$ by involving one or more of the similarity measurements: mutual information, Dice coefficient, overlap coefficient, cosine, or for example Jaccard coefficient
\[
J_c = \frac{|\Omega_{a_k}\cap \Omega_{a_l}|}{|\Omega_{a_k}|+|\Omega_{a_l}|-|\Omega_{a_k}\cap \Omega_{a_l}|} \in [0,1]
\]
In this concep of similarity, $B_{a_k}\cap B_{a_l} = {|\Omega_{a_k}\cap \Omega_{a_l}|}/({|\Omega_{a_k}|+|\Omega_{a_l}|-|\Omega_{a_k}\cap \Omega_{a_l}|})$ $= J_c$ such that $e_j\in E$ if $r_p>0$. However the behavior of $r_p (0\le r_p\le 1)$ depends on the behavior of $\Omega_{a_k}\subset\Omega$, $\Omega_{a_l}\subset\Omega$ and $\Omega_{a_k}\cap\Omega_{a_l}\subset\Omega$: $|\Omega_{a_k}|\le|\Omega_{a_l}|$ or $|\Omega_{a_k}|\ge|\Omega_{a_l}|$, $|\Omega_{a_k}\cap\Omega_{a_l}|\le |\Omega_{a_k}|$, and $|\Omega_{a_k}\cap\Omega_{a_l}|\le |\Omega_{a_l}|$. If another measurement concept is similar to $J_c$, then Theorem \ref{theo:behavior} is proved. Therefore, we have
\begin{corollary}
If the behavior of social actors behaves in clusters (of big data)
then the behavior of the clusters (of big data) can be represented by the extracted
social network.
\end{corollary}

\section{Conclusion}
In this social network study we have presented an analysis for formulating the behavior of resources of social network as a social network mining. Formulation based on a search engine model and the clustering model, and we have obtained an explanation that there are relations between social actors/vertices, relationships/edges, and documents/web based on the clusters are formed. The future work will involve the extraction of a social network to describe the research collaboration for exploring the behavior of social actors and their relationships.


\begin{thebibliography}{12}
\bibitem{nasution2010} Nasution, Mahyuddin K. M., and Noah, S. A. M., Superficial method for extracting social network for academic using Web snippets. In: Yu, J. et al. (eds.): Rough Set and Knowledge Technology (RSKT), LNAI, vol. 6401, 483-390, Springer, Heidelberg (2010).
\bibitem{arasu2003} Arasu, A., and Garcia-Molina, H., Extracting structured data from web pages. SIGMOD, CA, San Diego (2003).
\bibitem{nasution2012a} Nasution, Mahyuddin K. M., and Noah, S. A., A methodology to extract social network from the Web snippet. Cornell University Library, arXiv:1211.5877v1 [cs.SI] (26 Nov 2012).
\bibitem{bent2006} Bent, L., Rabinovich, M., Voelker, G. M. and Xiao, Z., Characterization of a large web site population with implications for content delivery. World Wide Web 9: 505-536 (2006).
\bibitem{abdillah2014} Abdillah, L. A., Indonesian's presidential social media campaigns. In Seminar Nasional Sistem Informasi Indonesia (SESINDO2014), ITS, Surabaya (2014).
\bibitem{nasution2011} Nasution, Mahyuddin K. M., and Noah, S. A. M., Extraction of academic social network from online database. In Shahrul Azman Mohd Noah et al. (eds.): Proceeding of 2011 International Conference on Semantic Technology and Information Retrieval (STAIRS'11), 64-69, IEEE, Putrajaya, Malaysia (2011).
\bibitem{nizamb2015} Nizam B., K., and Noah, S. A. M., Efficient identity matching using static pruning q-gram indexing approach. Decision Support Systems 73, 97-108 (2015).
\bibitem{yao2012} Yao, Y., Tong, H., Xu, F., and Lu, J., Subgraph extraction for trust inference in social networks. IEEE/ACM International Conference on Advances in Social Network Analysis and Mining, 163-170 (2012).
\bibitem{nasution2015} Nasution, Mahyuddin K. M., Elveny, M., Syah, R., and Noah, S. A. Behavior of the resources in the growth of social network. Proceedings of the 5th International Conference on Electrical Engineering and Informatics (ICEEI), 551-554, IEEE (2015).
\bibitem{scott2000} Scott, J. P., Social Network Analysis: A Handbook, 2nd ed., Sage Publications, London (2000).
\bibitem{memon2010} Nemon, N., Xu, J. J., Hicks, D. L., and Chen, H., Social network data mining: Research questions, techniques, and applications. In N. Memon et.al (eds.): Data Mining for Social Network Data, Annuals of Information Systems 12, London, Springer: 1-8 (2010).
\bibitem{nasution2012b} Nasution, Mahyuddin K. M., and Noah, S. A., Probabilistic Generative Model of Social Network Based on Web Features. Cornell University Library, arXiv:1207.3894v1 [math.PR] (17 Jul 2012).
\bibitem{nasution2011a} Nasution, Mahyuddin K. M., Kolmogorov Complexity: Clustering objects and similarity. Bulletin of Mathematics 3(1): 1-16 (2011).
\bibitem{wellman1996} Wellman, B., Salaff, J., Dimitrova, D., Garton, I., Gulia, M., and Haythornthwite, C., Computer networks as social networks: Collaborative work, telework, and virtual community. Ann. Rev. Social 22: 213-238 (1996).
\bibitem{nasution2012c} Nasution, Mahyuddin K. M., Simple search engine model: Adaptive properties. Cornell University Library, arXiv:1212.4702v1 [cs.IR] (19 Dec 2012).
\bibitem{nasution2012d} Nasution, Mahyuddin K. M., and Noah, S. A., Information retrieval model: A social network extraction perspective. In: Proceeding of IEEE International Conference on Information Retrieval \& Knowledge Management (CAMP12), Putrajaya-Malaysia (2012).
\bibitem{nasution2012e} Nasution, Mahyuddin K. M., Simple search engine model: Adaptive properties for doubleton. Cornell University Library, arXiv:1212.4702v1 [cs.IR] (19 Dec 2012).
\bibitem{chen2008} Chen, L.-C. and Sakai, A., Critical behavior and the limit distribution for long-range oriented percolation. I. Prob. Theory Relat. Fields 142: 151-188 (2008).
\bibitem{wang2011} Wang, X., Wang, X., Hu, C., He, K., Jiang, J., and Liu, B., Measurements on movie distribution behavior in peer-to-peer networks. 12th IFIP/IEEE IM (2011).
\bibitem{nasution2014a} Nasution, Mahyuddin K. M., Extracting keyword for disambiguating name based on the overlap principle. Proceeding of International Conference on Information Technology and Engineering Application (4-th ICIBA), Book 1, 119-125, February 20-21, (2015). (Cornell University Library, arXiv:1212.3012 [cs.IR] (30 Jan 2016).
\end{thebibliography}
\end{document}